\documentclass[journal]{IEEEtran}

\usepackage[tbtags]{amsmath} % defines many math commands and
\usepackage{amssymb}  % get, among others, blackboard bold fonts
\usepackage{enumerate}
\usepackage{amsfonts}
\usepackage[caption=false]{subfig}
\usepackage{color}
\usepackage{cite}

\usepackage{hyperref}
\usepackage{breakurl}

\usepackage{epsfig, graphicx, psfrag}

\usepackage{units}

\usepackage{mathtools}
\mathtoolsset{showonlyrefs=true}
\usepackage{amsthm}
\theoremstyle{plain}

\theoremstyle{plain}
\newtheorem{thm}{Theorem}
\newtheorem{lemma}{Lemma}

\newtheorem{corol}{Corollary}
\newtheorem{prop}{Proposition}

\theoremstyle{definition}
\newtheorem{defn}{Definition}

\theoremstyle{remark}
\newtheorem{remark}{Remark}

\newtheorem{example}{Example}

\newtheorem*{scheme*}{Scheme}

\newtheorem*{protocol*}{Protocol}
\providecommand{\thmref}[1]{Theorem~\ref{#1}}
\providecommand{\exref}[1]{Example~\ref{#1}}
\providecommand{\defnref}[1]{Definition~\ref{#1}}
\providecommand{\secref}[1]{Section~\ref{#1}}
\providecommand{\lemref}[1]{Lemma~\ref{#1}}
\providecommand{\propref}[1]{Proposition~\ref{#1}}
\providecommand{\remref}[1]{Remark~\ref{#1}}
\providecommand{\figref}[1]{Fig.~\ref{#1}}
\providecommand{\colref}[1]{Corollary~\ref{#1}}

\newcommand{\ie}{i.e.}
\newcommand{\eg}{e.g.}

\newcommand{\etal}{\emph{et al.}}

\newcommand{\bm}[1]{\mbox{\boldmath{$#1$}}}

\newcommand{\trace}{\mathrm{trace}}
\newcommand{\diag}{\mathrm{diag}}

\newcommand{\mH}{\mathcal{H}}

\newcommand{\mT}{\mathcal{T}}
\newcommand{\mU}{\mathcal{U}}
\newcommand{\mV}{\mathcal{V}}

\newcommand{\Comment}[1]{}
\newcommand{\old}[1]{}
\newcommand{\rem}[1]{}

\providecommand{\tN}{\tilde{N}}

\newcommand{\ZF}{{\textrm{ZF}}}

\newcommand{\hw}{\hat{w}}

\newcommand{\ty}{\tilde y}

\newcommand{\bd}{{\bf d}}

\newcommand{\bs}{\bm s}
\newcommand{\by}{\bm y}
\newcommand{\bT}{{\bf T}}
\newcommand{\bx}{{\bm x}}

\newcommand{\bz}{{\bm z}}

\newcommand{\bsigma}{{\bm \sigma}}

\providecommand{\tx}{\tilde{x}}
\providecommand{\tbx}{\tilde{\bx}}
\providecommand{\tby}{\tilde{\by}}

\providecommand{\bA}{{\bf A}}

\providecommand{\bG}{{\bf G}}
\providecommand{\bH}{{\bf H}}
\providecommand{\bI}{{\bf I}}
\providecommand{\bK}{{\bf K}}

\providecommand{\bU}{{\bf U}}
\providecommand{\bV}{{\bf V}}

\providecommand{\tbU}{\tilde{\bU}}

\newcommand{\cN}{{\mathcal N}}

\providecommand{\comment}[1]{}

\newcommand{\BC}{{\text{BC}}}
\newcommand{\MAC}{{\text{MAC}}}

\newcommand{\CS}{{\text{CS}}}
\newcommand{\PNC}{{\text{PNC}}}

\newcommand{\beqn}[1]{\begin{eqnarray}\label{#1}}
\newcommand{\eeqn}{\end{eqnarray}}
\newcommand{\beq}[1]{\begin{equation}\label{#1}}
\newcommand{\eeq}{\end{equation}}

\providecommand{\common}{\text{common}}
\providecommand{\Ccommon}{\mathrm{C}_\common}

\providecommand{\C}[1]{\mathrm{C} \left( #1 \right)}

\providecommand{\proper}{proper}
\providecommand{\Proper}{Proper}

\renewcommand{\dagger}{\mathit{T}}

\usepackage[mathcal]{euscript}
\usepackage{multirow}

\usepackage[T1]{fontenc}
\usepackage{frcursive}
\usepackage{calligra}
\newcommand{\setfont}[2]{{\fontfamily{#1}\selectfont #2}}
\newcommand{\sss}{\text{\setfont{frc}{s}}}
\newcommand{\xxx}{\text{\setfont{frc}{x}}}
\newcommand{\yyy}{\text{\setfont{frc}{y}}}
\newcommand{\zzz}{\text{\setfont{frc}{z}}}

\providecommand{\eps}{\epsilon}

\makeatletter
\newcommand{\vast}{\bBigg@{4}}
\newcommand{\Vast}{\bBigg@{4.9}}
\makeatother 

\providecommand{\HSNR}{\mathrm{HSNR}}
\providecommand{\mC}{\mathcal{C}}
\providecommand{\XOR}{\textsc{xor}}

\begin{document}
%%%%%%%%%%%%%%%%%%%%%%%%%%%%%%%%%%%%%%%%%%%%%%%%%%%%%%%%%%%%%%%%%%%%%%%%%%%%%%%%%%%%%%%%%%%%%%%%%%%%%%%%%%%%%%%%%%%%%
\title{The Dirty MIMO Multiple-Access Channel}

\author{
Anatoly Khina, 
Yuval Kochman, 
and 
Uri Erez
  \thanks{The material in this paper was presented in part at the \emph{2011 IEEE International Symposium of Information Theory}, St.~Petersburg, Russia, and at the \emph{2016 IEEE International Symposium of Information Theory}, Barcelona, Spain.}
  \thanks{A.~Khina was with the Department of Electrical Engineering-Systems, Tel Aviv University, Tel Aviv 69978, Israel.
        He is now with the Department of Electrical Engineering, California Institute of Technology, Pasadena, CA 91125, USA
        (e-mail: \nolinkurl{khina@caltech.edu}).}
  \thanks{Y.~Kochman is with the the Rachel and Selim Benin School of Computer Science and Engineering, Hebrew University of Jerusalem, Jerusalem 91904, Israel
        (e-mail: \mbox{\nolinkurl{yuvalko@cs.huji.ac.il}}).}
  \thanks{U.~Erez is with the Department of Electrical Engineering-Systems, Tel Aviv University, Tel Aviv 69978, Israel.
        (e-mail: \nolinkurl{uri@eng.tau.ac.il}).}
  \thanks{The work of A.~Khina was supported by the Feder Family Award, 
	  by the Trotsky Foundation at Tel Aviv University, 
	  by the Yitzhak and Chaya Weinstein Research Institute for Signal Processing, 
	  and by the Clore Israel Foundation.}
  \thanks{The work of Y.~Kochman was supported by the Israeli Science Foundation (ISF) under grant \#956/12, 
	  by the German--Israeli Foundation for Scientific Research and Development (GIF), 
	  and by the HUJI Cyber Security Research Center in conjunction with the Israel National Cyber Bureau in the Prime Minister's Office.}
  \thanks{The work of U.~Erez was supported, in part, by the ISF under Grant No.~1956/15.}
}

\maketitle

%%%%%%%%%%%%%%%%%%%%%%%%%%%%%%%%%%%%%%%%%%%%%%%%%%%%%%%%%%%%%%%%%%%%%%%%%%%%%%%%%%%%%%%%%%%%%%%%%%%%%%%%%%%%%%%%%%%%%%%%%

\begin{abstract}
    In the scalar dirty multiple-access channel, in addition to Gaussian noise, two additive interference signals are present, each known non-causally to a single transmitter. It was shown by Philosof \etal\ that for strong interferences,  an i.i.d.\ ensemble of codes does not achieve the capacity region. Rather, a structured-codes approach was presented, that was shown to be optimal in the limit of high signal-to-noise ratios, where the sum-capacity is dictated by the minimal (``bottleneck'') channel gain. In this paper, we consider the multiple-input multiple-output (MIMO) variant of this setting. In order to incorporate structured codes in this case, one can utilize matrix decompositions that transform the channel into effective parallel scalar dirty multiple-access channels. This approach however suffers from a ``bottleneck'' effect for each effective scalar channel and therefore the achievable rates strongly depend on the chosen decomposition. It is shown that a recently proposed decomposition, where the diagonals of the effective channel matrices are equal up to a scaling factor, is optimal at high signal-to-noise ratios, under an equal rank assumption. This approach is then extended to any number of transmitters. Finally, an application to physical-layer network coding for the MIMO two-way relay channel is presented.
\end{abstract}


\allowdisplaybreaks
%%%%%%%%%%%%%%%%%%%%%%%%%%%%%%%%%%%%%%%%%%%%%%%%%%%%%%%%%%%%%%%%%%%%%%%%%%%%%%%%%%%%%%%%%%%%%%%%%%%%%%%%%%%%%%%%%%%%%%%%%

\section{Introduction}
\label{s:intro}

The dirty-paper channel, first introduced by Costa \cite{Costa83}, is given by
\begin{align}
\label{eq:Costa}
    y = x + s + z ,
\end{align}
where $y$ is the channel output,
$x$ is the channel input subject to an average power constraint $P$,
$z$ is an additive white Gaussian noise (AWGN) of unit power,
and $s$ is an interference which is known non-causally to the transmitter but not to the receiver.

Costa \cite{Costa83} showed that the capacity of this channel, when the interference is i.i.d.\ and Gaussian,
is equal to that of an interference-free channel $\frac{1}{2} \log(1 + P)$, \ie, as if $s \equiv 0$.
This result was subsequently extended to ergodic interference in~\cite{CohenLapidoth02}
and to arbitrary interference in~\cite{ErezShamaiZamir05},
where to achieve the latter, a structured lattice-based coding scheme was used.

This model serves as an information-theoretic basis for the study of interference cancellation techniques,
and was applied to different network communication scenarios; see, \eg, \cite{ZamirBookNazerChapter}.

Its multiple-input multiple-output (MIMO) variant as well as its extension to MIMO broadcast with private messages can be easily treated either directly or via scalar dirty-paper coding (DPC) and an adequate orthogonal matrix decomposition, the most prominent being the singular-value decomposition (SVD) and the QR decomposition (QRD); see, \eg, \cite{GinisCioffiAsilomar,CaireShamai03,YuCioffiSumCapacity,WSS06,UCD}.

Philosof \etal~\cite{PhilosofZamirErezKishti09} extended the dirty-paper channel to the case of $K$ multiple (distributed) transmitters,
each transmitter, corresponding to a different user, knowing a different part of the interference: 
\begin{align}
\label{eq:DoublyDirtyMAC:Model}
    y = \sum_{k=1}^K \left( x_k + s_k \right) + z ,
\end{align}
where $y$ and $z$ are as before, $x_k$ ($k=1, \ldots, K$) is the input of transmitter $k$ and is subject to an average power constraint $P_k$,
and $s_k$ is an \emph{arbitrary} interference sequence which is known non-causally to transmitter $k$ but not to the other transmitters nor to the receiver.

\begin{figure}[t]
    \centering
    \epsfig{file = 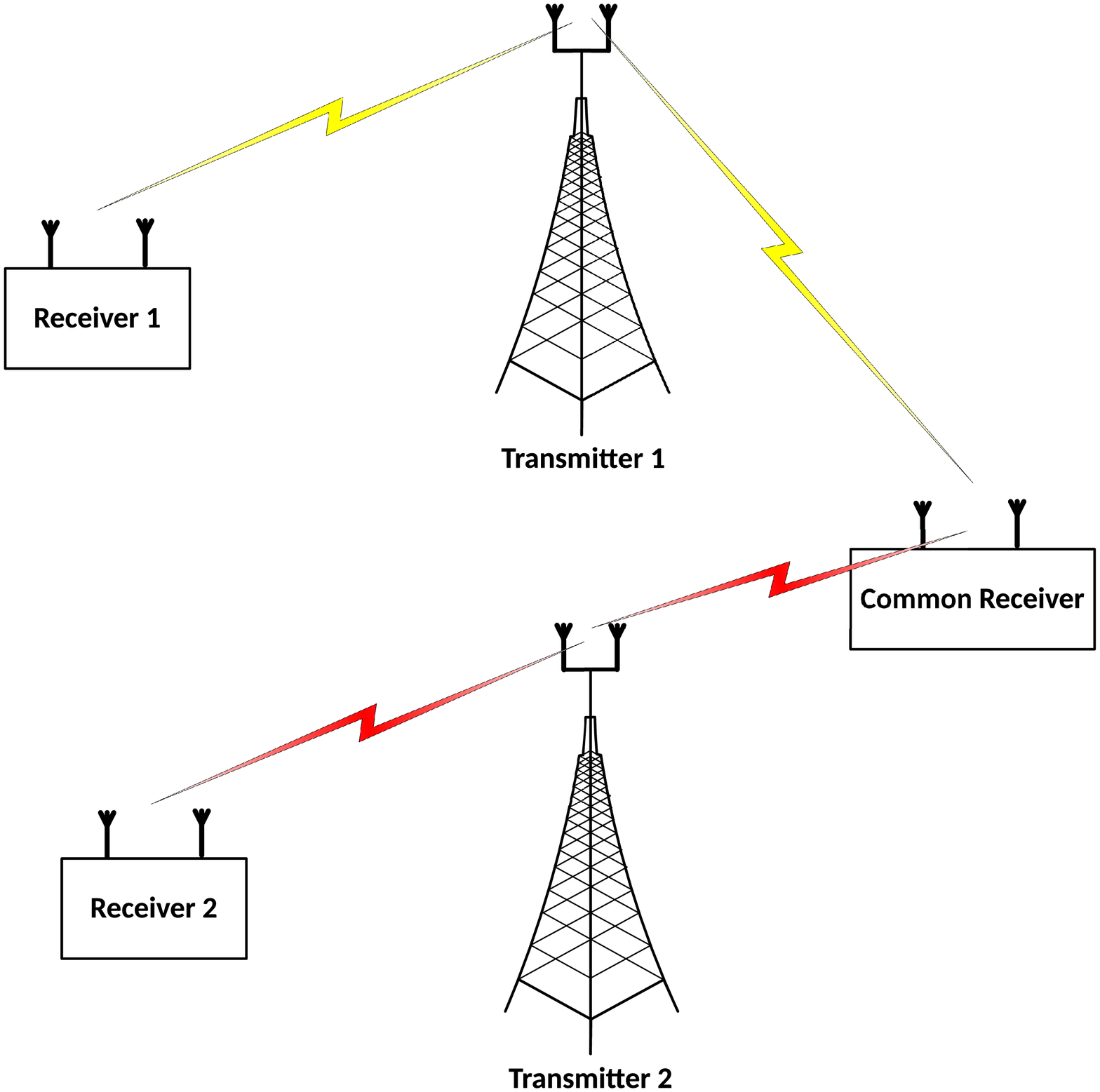, width = \columnwidth}
    \caption{A wireless dirty-paper channel scenario. The base stations communicate with the \emph{common receiver} over a multiple-access channel as well as with \emph{receiver 1} and \emph{receiver 2} that serve as interferences for the communication to the \emph{common receiver}.}
\label{fig:bus}
\end{figure}

This scenario is encountered in practice in cases where, for instance, non-cooperative base stations transmit data (over a multiple-access link) to a \emph{common receiver}
as well as to separate distinct receivers, which serve as interferences known at the transmitters for the communication to the common receiver. This scenario is described in \figref{fig:bus} for $K=2$ receivers.

The capacity region of this scenario, termed the \emph{dirty multiple-access channel} (DMAC) in \cite{PhilosofZamirErezKishti09},
was shown
to be contained (``outer bound'') in the region of all rate tuples $(R_1, \ldots, R_K)$ satisfying
\begin{align}
\label{eq:SISO:DMAC:UB}
    \sum_{k=1}^K R_k \leq \frac{1}{2} \log \left( 1 + \min_{k = 1, \ldots, K} P_k \right) ,
\end{align}
and
to contain (``achievable region'') all rate tuples $(R_1, \ldots, R_K)$ satisfying\footnote{In addition to \eqref{eq:SISO:DMAC:IB}, other inner bounds which are tighter in  certain cases are derived in \cite{PhilosofZamirErezKishti09}.}
\begin{align}
\label{eq:SISO:DMAC:IB}
    \sum_{k=1}^K  R_k \leq \frac{1}{2} \left[ \log \left( \frac{1}{K} + \min_{k = 1, \ldots, K} P_k\right) \right]^+,
\end{align}
where $[x]^+ \triangleq \max\{0,x\}$.
These two regions coincide in the limit of high signal-to-noise ratios (SNRs)~--- \mbox{$P_1,\ldots, P_K \gg 1$}~--- 
thus establishing the capacity region in this limit to be equal to the region of all rate tuples $(R_1, \ldots, R_K)$ satisfying
\begin{align}
\label{eq:SISO:DMAC:capacity}
    \sum_{k=1}^K R_k \leq \frac{1}{2} \log \left( \min_{k = 1, \ldots, K} P_k \right) .
\end{align}
That is, the sum-capacity suffers from a \emph{bottleneck problem} and reduces to the minimum of the individual capacities in this limit, where by the individual capacity of user $k$, we mean the capacity from transmitter $k$ to the receiver where all other transmitters are silent.
Interestingly, Costa's random binning technique does not achieve the rate region \eqref{eq:SISO:DMAC:IB} or the high-SNR region \eqref{eq:SISO:DMAC:capacity},
and structured lattice-based techniques need to be used \cite{PhilosofZamirErezKishti09,PhilosofZamir09}.

The 
MIMO counterpart of the problem is given by
\begin{align}
\label{eq:MIMO_DMAC:model:intro}
   \by = \sum_{k=1}^K \left( \bH_k \bx_k + \bs_k \right) + \bz .
\end{align}
For simplicity, we assume for now that all vectors are of equal length $N$.\footnote{We shall depart from this assumption later and treat the more general case of full-rank channel matrices where the number of receive antennas is larger or equal to that of the transmit antennas of each of the transmitters.} 
We further assume, without loss of generality, that the square channel matrices $\{\bH_k\}$ all have unit determinant, since any other value can be absorbed in $P_k$. 
The AWGN vector $\bz$ has i.i.d.\ unit-variance elements, while the interference vectors $\{\bs_k\}$ are arbitrary as in the scalar case. 
The transmitters are subject to average power constraints $\{P_k\}$.

In the high-SNR limit (where all powers satisfy $P_k \gg 1$), the individual capacity of the $k$-th user is given by:\footnote{The optimal covariance matrix in the limit of high SNR is white; see \lemref{lem:WaterFilling} in the sequel.}
\begin{align}
    \frac{N}{2} \log \left( \frac{P_k }{N} \right) .
\end{align}
Thus, similarly to the scalar case \eqref{eq:SISO:DMAC:capacity}, one can expect the high-SNR capacity region to be given by
\begin{align}
\label{eq:MIMO:bottleneck:intro}
    \sum_{k=1}^K R_k \leq  \frac{N}{2} \log \left( \frac{ \min\limits_{k = 1, \ldots, K} P_k }{N} \right) .
\end{align}

However, in contrast to the single-user setting~\eqref{eq:Costa}, the extension of the scalar DMAC to the MIMO case is not straightforward.
As structure is required even in the scalar case~\eqref{eq:DoublyDirtyMAC:Model},
one cannot use a vector random codebook. To overcome this, we suggest to employ $N_r$ parallel scalar schemes,
each using the lattice coding technique of \cite{PhilosofZamirErezKishti09}.
This is in the spirit of the capacity-achieving SVD \cite{Telatar99} or QRD \cite{Foschini96,CioffiForneyGDFE,Wolniansky_V-BLAST,HassibiVBLAST} based schemes, that were proposed for MIMO communications (motivated by implementation considerations). The total rate is split between multiple scalar codebooks, each one enjoying a channel gain according to the respective diagonal value of the equivalent channel matrix obtained by the channel decomposition.

Unfortunately, for the MIMO DMAC problem, neither the SVD nor the QRD is suitable, \ie,
their corresponding achievable rates
cannot approach~\eqref{eq:MIMO:bottleneck:intro}.
Applying the SVD is not possible in the MIMO DMAC setting as joint diagonalization with the same orthogonal matrix on one side
does not exist in general.
Applying the QRD to each of the orthogonal matrices, in contrast, is possible as it requires an orthogonal operation only at the transmitter.\footnote{More precisely, RQ decompositions need to be applied to the channel matrices in this case.}
However, the resulting matrices will have non-equal diagonals in general, corresponding to non-equal SNRs. Specifically, denoting the $i$-th diagonal element of the \mbox{$k$-th} matrix by $d_{k;i}$, the resulting high-SNR sum-rate would be limited to
\begin{align}
\label{eq:MIMO:bad_bottleneck}
    \sum_{k=1}^K R_k \leq \sum_{i=1}^{N_r} \frac{1}{2} \log \left( \frac{\min\limits_{k = 1, \ldots, K} \left( P_k d_{k;i}^2 \right)}{N_r} \right)
\end{align}
in this case.
As this represents a \emph{per-element} bottleneck, the rate is in general much lower than \eqref{eq:MIMO:bottleneck:intro}.

In this work 
we make use of a recently proposed joint orthogonal triangularization \cite{STUD:SP} to remedy the problem, i.e., to transform the per-element bottleneck \eqref{eq:MIMO:bad_bottleneck} into a global one as in \eqref{eq:MIMO:bottleneck:intro}. Specifically, the decomposition
allows to transform two matrices (with equal determinants) into triangular ones with \emph{equal diagonals}, using the same orthogonal matrix on the left~--- corresponding to a common operation carried at the receiver~--- and different orthogonal matrices on the right~--- corresponding to different operations applied by each of the transmitters.
The equal-diagonals property implies that the minimum in \eqref{eq:MIMO:bad_bottleneck} is not active and hence the per-element bottleneck problem, incurred in the QRD-based scheme, is replaced by the more favorable vector bottleneck \eqref{eq:MIMO:bottleneck:intro}.

The rest of the paper is organized as follows.
We start by introducing the channel model in \secref{s:model}.
We then present the
ingredients we use: the matrix decomposition is presented in \secref{s:decompositions}, and a structured coding scheme for the single-user ``dirty'' MIMO channel is presented in \secref{s:DPC}. Our main result, the high-SNR capacity of the two-user MIMO DMAC \eqref{eq:MIMO_DMAC:model:intro} is given in \secref{s:MIMO:DMAC:2users}, using a structured scheme.
We extend this result to the $K$-user case in \secref{s:MIMO:DMAC:Kusers}.
We then demonstrate the usefulness of the proposed technique for MIMO physical-layer network coding
in \secref{s:TWRC},
by constructing a scheme that achieves capacity in the limit of high SNR for the MIMO two-way relay channel. 
We conclude the paper in \secref{s:discussion}.

%%%%%%%%%%%%%%%%%%%%%%%%%%%%%%%%%%%%%%%%%%%%%%%%%%%%%%%%%%%%%%%%%%%%%%%%%%%%%%%%%%%%%%%%%%%%%%%%%%%%%%%%%%%%%%%%%%%%%%%%%

\section{Problem Statement}
\label{s:model}

The $K$-user MIMO DMAC is given by:
\begin{align}
\label{eq:MIMO_DMAC:model}
   \by = \sum_{k=1}^K \left( \bH_k \bx_k + \bs_k \right) + \bz ,
\end{align}
where $\by$ is the channel output vector of length $N_r$,\footnote{All vectors in this paper are assumed column vectors.}
$\bx_k$ ($k=1, \ldots, K$) is the input vector of transmitter $k$ of length $N_{t;k}$ and is subject to an average power constraint $P_k$ defined formally in the sequal,
$\bz$ is an AWGN vector with an identity covariance matrix,
and $\bs_k$ is an interference vector of length $N_r$ which is known non-causally to transmitter $k$ but not to the other transmitters nor to the receiver.
The interference vector signals $\{\bs_k\}$ are assumed to be arbitrary sequences.
We consider a closed-loop scenario, meaning that the $N_r \times N_{t;k}$ channel matrix $\bH_k$ is known everywhere and that it satisfies the following properties.
\begin{defn}[\Proper] 
\label{def:proper}
    A matrix $\bH$ of dimensions $N_r \times N_t$ is said to be \proper\ if it has no fewer columns than rows, \ie, $N_r \leq N_t$,
    is full rank (namely of rank $N_r$)
    and satisfies
    \begin{align}
    \label{eq:DPC:MIMO:det=1}
        \det\left( \bH \bH^\dagger \right) = 1.
    \end{align}
\end{defn}

\begin{remark}
    Similarly to the special case of full-rank matrices discussed in the introduction, 
    a full-rank $N_r \times N_t$ matrix $\bH$ with $N_r \leq N_t$ and 
    $\det\left( \bH \bH^\dagger \right) = a \neq 1$ may always be normalized to satisfy \eqref{eq:DPC:MIMO:det=1} by absorbing $a$ in the associated power constraint.
\end{remark}

Transmission is carried out in blocks of length $n$.
The input signal transmitted by transmitter $k$ is given by
\begin{align}
    \bx_k^n = f_k\left( w_k, \bs_k^n \right) ,
\end{align}
where we denote by $a^n$ blocks of $a$ at time instants $1,2,\ldots,n$, i.e., $a^n=a[1], \ldots, a[n]$,
$w_k$ is the conveyed message by this user which is chosen uniformly from $\left\{ 1, \ldots, \left\lceil 2^{n R_k} \right\rceil \right\}$, $R_k$ is its transmission rate, and $f_k$ is the encoding function.
The input signal $\bx_k$ is subject to an average power constraint
\begin{align}
    \frac{1}{n} \sum_{\ell=1}^n \bx_{k}^2[\ell] \leq P_k .
\end{align}

The receiver reconstructs the messages $\hw_1, \ldots \hw_K$ from the channel output, using a decoding function $g$:
\begin{align}
    \left( \hw_1, \ldots, \hw_K \right) = g\left( \by^n \right) .
\end{align}

A rate tuple $\left( R_1, \ldots, R_K \right)$ is said to be achievable if for any $\eps > 0$, however small,
there exist $n$, $f$ and $g$, such that the error probability is bounded from the above by $\eps$:
\begin{align}
    \Pr\left( \hw_1 \neq w_1, \ldots, \hw_K \neq w_K \right) \leq \eps .
\end{align}

The capacity region is defined as the closure of all achievable rate tuples.

%%%%%%%%%%%%%%%%%%%%%%%%%%%%%%%%%%%%%%%%%%%%%%%%%%%%%%%%%%%%%%%%%%%%%%%%%%%%%%%%%%%%%%%%%%%%%%%%%%%%%%%%%%%%%%%%%%%%%%%

\section{Background: Orthogonal Matrix Triangularization}
\label{s:decompositions}

In this section we briefly recall some important matrix decompositions that will be used in the sequel.
In \secref{ss:GTD} we recall the generalized triangular decomposition (GTD)
and some of its important special cases.
Joint orthogonal triangularizations of two matrices are discussed in \secref{ss:STUD}.

\subsection{Single Matrix Triangularization}
\label{ss:GTD}

 Let $\bA$ be a proper matrix (recall \defnref{def:proper}) of dimensions $M \times N$. A generalized triangular decomposition (GTD) of $\bA$ is given by:
     \begin{align}
        \label{eq:GTD}
            \bA &= \bU \bT \bV^\dagger ,
        \end{align}
where   $\bU$ and $\bV$ are orthogonal matrices of dimensions \mbox{$M \times M$} and $N \times N$, respectively, and $\bT$ is a generalized lower-triangular matrix:
\begin{align}
    T_{i,j} &= 0   \,, & \forall i < j \,.
\end{align}
Namely, it has the following structure:
\begin{align}
    \bT =
    {\scriptstyle M} \Vast\{
    \overbrace{
    \begin{bmatrix}
        T_{1,1}  & 0       & 0       & \cdots & 0 & \cdots & 0
     \\ T_{2,1}  & T_{2,2}  & 0       & \cdots & 0 & \cdots & 0
     \\ \vdots  & \vdots  & \ddots  & \      & \vdots & \ddots &\vdots      & \
     \\ T_{M,1} & T_{M,2} & \cdots & T_{M,M} & 0 & \cdots & 0
    \end{bmatrix}
    }^{N}
    \,.
\end{align}

The diagonal entries of $\bT$ always have a unit product.\footnote{Since $\bA$ is a proper matrix it is full rank by definition; thus, all the diagonal values of $\bT$, $\{T_{i,i}\}$, are non-zero.} 
Necessary and sufficient conditions for the existence of a GTD for a prescribed diagonal $\{T_{i,i}\}$ are known, along with explicit constructions of such a decomposition \cite{WeylCondition,WeylConditionInverse_ByHorn,GTD,QRS-GTD,UnityTriangularization}.

The following three important special cases of the GTD are well known; all of them are guaranteed to exist for a proper matrix $\bA$.\footnote{See \cite{JET:SeveralUsers2015:FullPaper} for a geometrical interpretation of these decompositions.}
\subsubsection{SVD (see, \eg, \cite{GolubVanLoan3rdEd})}
Here, the resulting matrix $\bT$ in \eqref{eq:GTD} is a \emph{diagonal} matrix,
and its diagonal elements are equal to the singular values of the decomposed matrix $\bA$.

\subsubsection{QR Decomposition (see, \eg, \cite{GolubVanLoan3rdEd})}
In this decomposition,
the matrix
$\bV$ in \eqref{eq:GTD} equals the identity matrix and hence does not depend on the matrix $\bA$.
This decomposition can be constructed by performing Gram--Schmidt orthonormalization on the (ordered) columns of the
matrix $\bA$.

\subsubsection{GMD (see \cite{UnityTriangularization,QRS,GMD})}
\label{sss:GMD}
The diagonal elements of $\bT$ in this decomposition are all equal
to the geometric mean of its singular values $\bsigma(\bA)$,
which is real and positive.

%------------------------------------------------------------------------------------------------------------------------

\subsection{Joint Matrix Triangularization}
\label{ss:STUD}

  Let $\bA_1$ and $\bA_2$ be two proper matrices of dimensions \mbox{$M_1 \times N$} and \mbox{$M_2 \times N$}, respectively.
    A joint triangularization of  these two matrices is given by:
    \begin{subequations}
    \label{eq:STUD}
    \noeqref{eq:STUD:A1,eq:STUD:A2}
    \begin{align}
        \bA_1 &= \bU_1 \bT_1 \bV^\dagger ,
    \label{eq:STUD:A1}
     \\ \bA_2 &= \bU_2 \bT_2 \bV^\dagger ,
    \label{eq:STUD:A2}
    \end{align}
    \end{subequations}
    where $\bU_1$, $\bU_2$ and $\bV$ are orthogonal matrices of dimensions \mbox{$M_1 \times M_1$}, $M_2 \times M_2$ and $N \times N$, respectively, and $\bT_1$ and $\bT_2$ are generalized lower-triangular matrices of dimensions $M_1 \times N$ and $M_2 \times N$, respectively.

It turns out that the existence of such a decomposition depends on the diagonal \emph{ratios} $\{T_{1;ii}/T_{2;ii}\}$. Necessary and sufficient conditions were given in~\cite{STUD:SP}. Specifically, it was shown that there always exists a decomposition with unit ratios, i.e.,
\begin{align}
    T_{1;ii} &= T_{2;ii} \,, &&              && i = 1, \ldots, N \,.
\end{align}
Such a decomposition is coined the joint equi-diagonal triangularization (JET).\footnote{See \cite{JET:SeveralUsers2015:FullPaper} for a geometrical interpretation of the JET.}
Technically, the existence of JET is an extension of
the existence of the (single-matrix) GMD.

Unfortunately, JET of more than two matrices is in general not possible  \cite{JET:SeveralUsers2015:FullPaper}.
Nonetheless, in Section~\ref{s:MIMO:DMAC:Kusers} we present a way to overcome this obstacle.

%%%%%%%%%%%%%%%%%%%%%%%%%%%%%%%%%%%%%%%%%%%%%%%%%%%%%%%%%%%%%%%%%%%%%%%%%%%%%%%%%%%%%%%%%%%%%%%%%%%%%%%%%%%%%%%%%%%%%%%

\section{Background: Single-User MIMO Dirty-Paper~Channel}
\label{s:DPC}

In this section we review the (single-user) MIMO dirty-paper channel, corresponding to setting $K=1$ in \eqref{eq:MIMO_DMAC:model}:
\begin{align}
\label{eq:DPC:channel}
    \by = \bH \bx + \bs + \bz .
\end{align}
We suppress the user index of $\bx$, $\bs$, $\bH$ and $N_t$ in this case.

For an i.i.d.\ Gaussian interference vector, a straightforward extension of Costa's random binning scheme achieves the capacity of this channel,
\begin{align}
\label{eq:DPC:MIMO:Capaciy}
    \C{\bH, P} \triangleq \max_{\bK:\, \trace\left( \bK \right) \leq P} \frac{1}{2} \log \left| \bI_{N_r} + \bH \bK \bH^\dagger \right| ,
\end{align}
which is, as in the scalar case, equal to the interference-free capacity.
In the high-SNR limit, we have the following.
\begin{lemma}[See, \eg, \cite{MartinianWaterfilling}]
\label{lem:WaterFilling}
    The capacity of the single-user MIMO dirty-paper channel \eqref{eq:DPC:channel}
    satisfies
    \begin{align}
        \lim_{P\rightarrow\infty} [\mathrm{C}-R_\HSNR] = 0 ,
    \end{align}
    where
    \begin{align}
    \label{eq:DPC:MIMO:HighSNR:capacity}
       R_\HSNR  \triangleq \frac{N_r}{2} \log \frac{P}{N_r} \,.
    \end{align}
    Furthermore, this rate can
    be achieved by the input covariance matrix
    \begin{align}
    \label{eq:DPC:MIMO:HighSNR:CovMat}
        \bK = \frac{P}{N_t} \bI_{N_t} \,.
    \end{align}

\end{lemma}

The Costa-style scheme for the MIMO dirty-paper channel suffers from two major drawbacks. First,
it requires vector codebooks of dimension $N_t$, which depend on the specific channel $\bH$. And second, it does not admit an arbitrary interference. Both can be resolved by using the orthogonal matrix decompositions of \secref{s:decompositions} to reduce the coding task to that of coding for the \emph{scalar} dirty-paper channel \eqref{eq:Costa}.
For each scalar channel, the interference consists of two parts: a linear combination of the elements of the ``physical interference'' $\bs$ and a linear combination of the off-diagonal elements of the triangular matrix which also serves as ``self interference''.
When using the lattice-based scheme of \cite{ErezShamaiZamir05}, the capacity \eqref{eq:DPC:MIMO:Capaciy} is achieved even for an arbitrary interference sequence $\bs$.

\begin{scheme*}[Single-user zero-forcing MIMO DPC]
\ \\ \indent
  \textbf{Offline:}
  \begin{itemize}
  \item
    Apply any orthogonal matrix triangularization \eqref{eq:GTD} to the channel matrix $\bH$, 
    to obtain the orthogonal matrices $\bU$ and $\bV$, of dimensions $N_r$ and $N_t$, respectively, and the $N_r \times N_t$ generalized lower-triangular matrix $\bT$.
  \item
    Denote the vector of the diagonal entries of $\bT$ by \mbox{$\bd \triangleq \diag\left( \bT \right)$}.

    \item
	Construct $N_r$ good unit-power scalar dirty-paper codes with respect to SNRs $\left\{ d_i^2 P / N_r \right\}$.
    \end{itemize}

    \textbf{Transmitter:}
    At each time instant:
    \begin{itemize}
    \item
	Generates $\{ \tx_i \}$ in a successive manner from first ($i=1$) to last ($i=N_r$), 
    where $\tx_i$ is the corresponding entry of the codeword of sub-channel $i$, 
    the interference over this sub-channel is equal to
	\begin{align}
	    \sum_{\ell=1}^{i-1} T_{i,\ell} \tx_{\ell} + \sum_{\ell = 1}^{N_r} V_{i,\ell} s_\ell  ,
	\end{align}
	and $T_{i,\ell}$ and $V_{i,\ell}$ denote the $(i,\ell)$ entries of $\bT$ and $\bV$, respectively.
    \item
	Forms $\tbx$ with its first $N_r$ entries being $\{ \tx_i \}$ followed by $(N_t - N_r)$ zeros.
    \item
	Transmits $\bx$ which is formed by multiplying $\tbx$ by~$\bV$:
      \begin{align}
	  \bx = \bV \tbx .
      \end{align}
    \end{itemize}

  \textbf{Receiver:}

    \begin{itemize}
    \item
      At each time instant forms 
          $\tby = \bU^\dagger \by$.

    \item
      Decodes the codebooks using dirty-paper decoders, where $\tx_i$ is decoded from $\ty_{i}$.
  \end{itemize}
\end{scheme*}

As is well known, the zero-forcing (ZF) DPC scheme approaches capacity for proper channel matrices in the limit of high SNR.
This is formally stated as a corollary of \lemref{lem:WaterFilling}.

\begin{corol}
\label{thm:DPC:MIMO:ZF}
    For any proper channel matrix $\bH$, the ZF MIMO DPC scheme achieves $R_\HSNR$ \eqref{eq:DPC:MIMO:HighSNR:capacity}. Thus, it approaches the capacity of the MIMO dirty-paper channel \eqref{eq:DPC:channel}
    in the limit $P \rightarrow \infty$.
\end{corol}

\begin{proof}
    The ZF MIMO DPC scheme achieves a rate of
    \begin{align}
        R_\ZF &= \sum_{i=1}^{N_r} \frac{1}{2} \log \left( 1 + \frac{P}{N_r} d_i^2 \right)
     \\ &\geq \sum_{i=1}^{N_r} \frac{1}{2} \log \left( \frac{P}{N_r} d_i^2 \right)
     \\ &= N_r \cdot \frac{1}{2} \log \left( \frac{P}{N_r} \right) + \frac{1}{2} \log \left( \prod_{i=1}^{N_r} d_i^2 \right)
     \\ &= \frac{N_r}{2} \log \left( \frac{P}{N_r} \right)
        ,
    \end{align}
    where the last equality follows from \eqref{eq:DPC:MIMO:det=1}.
\end{proof}

\begin{remark}
\label{rem:MMSE_DPC}
    A minimum mean square error (MMSE) variant of the scheme achieves capacity
    for any SNR and any channel matrix (not necessarily \proper); see, \eg, \cite{UCD}.
    Unfortunately, extending the MMSE variant of the scheme to the DMAC setting is not straightforward,
    and therefore we shall concentrate on the ZF variant of the scheme.
\end{remark}

%%%%%%%%%%%%%%%%%%%%%%%%%%%%%%%%%%%%%%%%%%%%%%%%%%%%%%%%%%%%%%%%%%%%%%%%%%%%%%%%%%%%%%%%%%%%%%%%%%%%%%%%%%%%%%%%%%%%%%%

\section{Two-User MIMO DMAC}
\label{s:MIMO:DMAC:2users}

In this section we derive outer and inner bounds on the capacity region of the two-user MIMO DMAC \eqref{eq:MIMO_DMAC:model}.
We show that the two coincide for proper channel matrices in the limit of high SNRs.

The following is a straightforward adaptation of the outer bound of \cite{PhilosofZamirErezKishti09} for the scalar case \eqref{eq:SISO:DMAC:UB} to the two-user MIMO setting~\eqref{eq:MIMO_DMAC:model}.
It is formally proved in the Appendix.

\begin{prop}[Two-user sum-capacity outer bound]
\label{prop:DMAC:2user:UB}
    The sum-capacity of the two-user MIMO DMAC~\eqref{eq:MIMO_DMAC:model}
    is bounded from above by the minimum of the individual capacities:
    \begin{align}
    \label{eq:SumCapacityUB}
        R_1 + R_2 \leq \frac{1}{2} \log \min_{k=1,2} \
        \max_{\bK_k: \, \trace(\bK_k) \leq P_k} \left| \bI + \bH_k \bK_k \bH_k^\dagger \right| . \
    \end{align}
\end{prop}

We next introduce an inner bound that approaches the upper bound \eqref{eq:SumCapacityUB} in the limit of high SNRs.

\begin{thm}
\label{thm:proper_dim:IB:2user}
    For the two-user MIMO DMAC~\eqref{eq:MIMO_DMAC:model} with any \proper\ channel matrices $\bH_1$ and $\bH_2$,
    the region of all non-negative rate pairs $(R_1,R_2)$ satisfying
    \begin{align}
    \label{eq:proper_dim:IB:2user}
        R_1 + R_2 \leq \frac{N_r}{2} \left[ \log \left( \frac{\min\{P_1, P_2\}}{N_r} \right) \right]^+
    \end{align}
    is achievable.
\end{thm}

We give a constructive proof, employing a scheme that uses the JET of \secref{ss:STUD} to translate the two-user MIMO DMAC \eqref{eq:MIMO_DMAC:model}
into parallel SISO DMACs with equal channel gains (corresponding to equal diagonals).
As explained in the introduction, this specific choice of decomposition is essential.

\begin{scheme*}[Two-user MIMO DMAC]
\ \vspace{.2\baselineskip}

    \textbf{Offline:}
    \begin{itemize}
    \item
	Apply the JET of \secref{ss:STUD} to the channel matrices $\bH_1$ and $\bH_2$,
	to obtain the orthogonal matrices $\bU$, $\bV_1$, and $\bV_2$,
	of dimensions $N_r$, $N_{t;1}$ and $N_{t;2}$, respectively,
	and the generalized lower-triangular matrices $\bT_1$ and $\bT_2$ of dimensions $N_r \times N_{t;1}$ and $N_r \times N_{t;2}$, respectively.
    \item
	Denote the $N_r$ diagonal elements of $\bT_1$ and $\bT_2$ by $\{d_i\}$ (which are equal for both matrices).
    \item
	Construct $N_r$ good unit-power scalar DMAC codes with respect to SNR pairs $\{ (d_i^2 P_1 / N_r, d_i^2 P_2 / N_r) \}$.
    \end{itemize}

    \textbf{Transmitter $k$ ($k=1,2$):}
    At each time instant:
    \begin{itemize}
    \item
	Generates $\{ \tx_{k;i} \}$ in a successive manner from first ($i=1$) to last ($i=N_r$), 
    where $\tx_{k;i}$ is the corresponding entry of the codeword of user $k$ over sub-channel $i$, 
    the interference over this sub-channel is equal to
	\begin{align}
	    \sum_{\ell=1}^{i-1} T_{k;i,\ell} \tx_{k;\ell} + \sum_{\ell=1}^{N_r} V_{k;i,\ell} s_{k;\ell} ,
	\end{align}
	$T_{k;i,\ell}$ and $V_{k;i,\ell}$ are the $(i,\ell)$ entries of $\bT_k$ and $\bV_K$, respectively, and $s_{k;\ell}$ is the $\ell$-th entry of $\bs_k$.
    \item
	Forms $\tbx_k$ with its first $N_r$ entries being $\{ \tx_{k;i} \}$ followed by $(N_{t;k} - N_r)$ zeros.
    \item
	Transmits $\bx_k$ which is formed by multiplying $\tbx$ by~$\bV_k$:
      \begin{align}
	  \bx_k = \bV_k \tbx_k \,.
      \end{align}
    \end{itemize}

    \textbf{Receiver:}
    \begin{itemize}
    \item
      At each time instant, forms $\tby$ according to:
      \begin{align}
          \tby = \tbU^\dagger \by .
      \end{align}

    \item
      Decodes the codebooks using the decoders of the scalar DMAC codes, where $\tx_{1,i}$ and $\tx_{2;i}$ are decoded from $\ty_i$.
  \end{itemize}
\end{scheme*}

We use this scheme for the proof of the theorem.

\begin{proof}[Proof of \thmref{thm:proper_dim:IB:2user}]
    The proposed scheme achieves any rate pair $(R_1, R_2)$ whose sum-rate is bounded from below by
\vspace{-\baselineskip}
    \begin{subequations}
    \label{eq:IB:proper_dim}
    \noeqref{eq:IB:proper_dim:subchannels,eq:IB:proper_dim:explicit_rates,eq:IB:proper_dim:drop_half,eq:IB:proper_dim:split,eq:IB:proper_dim:highSNR_capacity}
    \begin{align}
        R_1 + R_2 &= \sum_{i=1}^{N_r} \left( r_{1;i} + r_{2;i} \right)
    \label{eq:IB:proper_dim:subchannels}
     \\ &\geq \sum_{i=1}^{N_r} \frac{1}{2} \left[ \log  \left( \frac{1}{2} + d_i^2 \cdot \frac{\min \{P_1, P_2\}}{N_r} \right) \right]^+
    \label{eq:IB:proper_dim:explicit_rates}
     \\ &\geq \sum_{i=1}^{N_r} \frac{1}{2} \log \left( d_i^2 \cdot \frac{\min \{P_1, P_2\}}{N_r} \right)
    \label{eq:IB:proper_dim:drop_half}
     \\ &= \sum_{i=1}^{N_r} \frac{1}{2} \log d_i^2 + \sum_{i=1}^{N_r} \frac{1}{2} \log \left( \frac{\min \{P_1, P_2\}}{N_r} \right) \ \ \ \ \
    \label{eq:IB:proper_dim:split}
     \\ &= \frac{N_r}{2} \log \left( \frac{\min \{P_1, P_2\}}{N_r} \right) ,
    \label{eq:IB:proper_dim:highSNR_capacity}
    \end{align}
    \end{subequations}
    where $r_{k;i}$ is the achievable rate of transmitter $k$ ($k=1,2$) over sub-channel $i$ ($i=1,\ldots, N_r$),
    \eqref{eq:IB:proper_dim:explicit_rates} follows from \eqref{eq:SISO:DMAC:IB},
    and \eqref{eq:IB:proper_dim:highSNR_capacity} holds true due to \eqref{eq:DPC:MIMO:det=1}.
\end{proof}

By comparing \propref{prop:DMAC:2user:UB} with \thmref{thm:proper_dim:IB:2user} in the limit of high SNR,
(recall \lemref{lem:WaterFilling}),
the following corollary follows.

\begin{corol}
\label{corol:MIMO_DMAC:2user:capacity}
    The capacity region of the two-user MIMO DMAC \eqref{eq:MIMO_DMAC:model} with any proper channel matrices $\bH_1$ and $\bH_2$ is given by $\mC_\HSNR+o(1)$, where $\mC_\HSNR$ is given by all rate pairs satisfying:
   \begin{align}
        R_1+R_2 \leq \frac{N_r}{2} \log \left( \frac{\min \{P_1, P_2\}}{N_r} \right)
    \end{align}
    and $o(1)$ vanishes as
    $\min\{ P_1, P_2 \} \to \infty$.
  \end{corol}

\begin{remark} \label{remark:finite_SNR} At any finite SNR, the scheme can achieve rates outside $\mC_\HSNR$. Specifically, the inequality  \eqref{eq:IB:proper_dim:drop_half} is strict, unless the achievable sum-rate is zero. However, in that case the calculation depends upon the exact diagonal values $\{d_i\}$; we do not pursue this direction. \end{remark}

%%%%%%%%%%%%%%%%%%%%%%%%%%%%%%%%%%%%%%%%%%%%%%%%%%%%%%%%%%%%%%%%%%%%%%%%%%%%%%%%%%%%%%%%%%%%%%%%%%%%%%%%%%%%%%

\section{$K$-User MIMO DMAC}
\label{s:MIMO:DMAC:Kusers}

In this section we extend the results obtained in \secref{s:MIMO:DMAC:2users} to MIMO DMACs with $K>2$ users.

The outer bound is a straightforward extension of the two-user case of \propref{prop:DMAC:2user:UB}.

\begin{prop}[$K$-user sum-capacity outer bound]
\label{prop:DMAC:Kuser:UB}
    The sum-capacity of the $K$-user MIMO DMAC~\eqref{eq:MIMO_DMAC:model}
    is bounded from the above by the minimum of the individual capacities:
    \begin{align}
        \sum_{k=1}^K R_k \leq \frac{1}{2} \log \min_{k = 1, \ldots, K} \ 
        \max_{\bK_k: \, \trace(\bK_k) \leq P_k} \left| \bI + \bH_k \bK_k \bH_k^\dagger \right| .
    \nonumber
    \end{align}
\end{prop}
For an inner bound, we would have liked to use a JET of $K>2$ matrices. As such a decomposition does not exist in general, we present a ``workaround'', following  \cite{JET:SeveralUsers2015:FullPaper}.

We process jointly $N$ channel uses and consider them as one time-extended channel use.
The corresponding time-extended channel is
\begin{align}
    \yyy = \sum_{k=1}^K \left( \mH_k \xxx_k + \sss_k \right) + \zzz \,,
\end{align}
where $\yyy$, $\xxx_k$, $\sss_k$, $\zzz$ are the time-extended vectors composed of $N$ ``physical'' (concatenated)
output, input, interference and noise vectors, respectively. The corresponding time-extended matrix $\mH_k$ is a block-diagonal matrix
whose $N$ blocks are all equal to $\bH_k$:
\begin{align}
\label{eq:time-extension:matrices}
    \mH_k = \bI_N \otimes \bH_k \,,
\end{align}
where $\otimes$ denotes the Kronecker product operation (see, \eg, \cite[Ch.~4]{HornJohnsonBook}).
As the following result shows, for such block-diagonal matrices we can achieve equal diagonals, up to edge effects that  can be made arbitrarily small by utilizing a sufficient number of time extensions $N$.

\begin{thm}[$K$-JET with edge effects \cite{JET:SeveralUsers2015:FullPaper}]
\label{thm:K-GMD}
    Let $\bH_1, \ldots, \bH_K$ be $K$ proper matrices of dimensions $N_r \times N_{t;1}, \ldots, N_r \times N_{t;K}$, respectively,
    and construct their time-extended matrices with $N$ blocks, $\mH_1, \ldots, \mH_K$, respectively, according to \eqref{eq:time-extension:matrices}.
    Denote
    \begin{align}
    \label{eq:tN}
	\tN =  N - N_r^{K-2} + 1 . 
    \end{align}
    Then, there exist matrices  $\mU, \mV_1, \ldots, \mV_K$ with orthonormal columns of dimensions 
    $N_r N \times N_r \tN, N_{t;1} N \times N_r \tN,  \ldots,$ $N_{t;K} N \times N_r \tN$, 
    respectively, such that
    \begin{align}
        \mT_k &= \mU^\dagger \mH_k \mV_k \,, & k = 1, \ldots, K \,,
    \end{align}
    where $\mT_1, \ldots, \mT_K$ are lower-triangular matrices of dimensions \mbox{$N_r \tN \times N_r \tN$}
    with equal diagonals of unit product.
\end{thm}

    The value of $\tilde{N}$ in \eqref{eq:tN}
    stems from a matrix truncation operation where we omit rows for which triangularization is not guaranteed; see, \eg, \cite[Section VII-B]{JET:SeveralUsers2015:FullPaper} for an example of the case $K = 3$, $N_r = N_{t;1} = N_{t;2} = N_{t;3} = 2$ and general $N$.

Since the ratio between $\tN$ and $N$ approaches 1 when $N$ goes to infinity:
    \begin{align} \label{eq:tN_limit}
        \lim_{N \to \infty} \frac{\tN}{N} = 1 ,
    \end{align}
the loss in rate due to the disparity between them goes to zero.
This is stated formally in the following theorem.

\begin{thm}
\label{thm:proper_dim:IB:Kuser}
    For the $K$-user MIMO DMAC~\eqref{eq:MIMO_DMAC:model},
    the region of all non-negative rate tuples $(R_1, \ldots, R_K)$ satisfying
    \begin{align}
    \label{eq:proper_dim:IB:Kuser}
        \sum_{k=1}^K R_k \leq \frac{N_r}{2} \left[ \log \left( \frac{ \min\limits_k P_k }{N_r} \right) \right]^+
    \end{align}
    is achievable.
\end{thm}

\begin{proof}
Fix some large enough $N$. Construct the channel matrices $\mH_1,\ldots,\mH_K$ as in  \eqref{eq:time-extension:matrices}, and set $\mU$, $\mV_1,\ldots,\mV_K$ and $\mT_1,\ldots,\mT_K$ according to \thmref{thm:K-GMD}. Now, over $nN$ consecutive channel uses, apply the  natural extension of the scheme of \secref{s:MIMO:DMAC:2users} to $K$ users, replacing $\{U_k\}, V, \{T_k\}$ with the obtained matrices. As in the proof of \thmref{thm:proper_dim:IB:2user}, we can attain any rate approaching
\begin{align}
    \sum_{k=1}^K R_k \leq \frac{N_r \tN}{2} \log \frac {\min\limits_{k=1, \ldots, K} P_k}{N_r} \,.
\end{align}
As we used the channel $nN$ times, we need to divide these rates by $N$; by \eqref{eq:tN_limit} the proof is complete.
\end{proof}

By comparing \propref{prop:DMAC:Kuser:UB} with \thmref{thm:proper_dim:IB:Kuser} in the limit of high SNR,
we can extend Corollary~\ref{corol:MIMO_DMAC:2user:capacity} as follows.

\begin{corol}
\label{corol:MIMO_DMAC:Kuser:capacity}
    The capacity region of the $K$-user MIMO DMAC~\eqref{eq:MIMO_DMAC:model} with any proper channel matrices $\bH_1,\ldots,\bH_K$ is given by $\mC_\HSNR+o(1)$, where $\mC_\HSNR$ is given by all rate pairs satisfying:
   \begin{align}
        \sum_{k=1}^K R_k \leq \frac{N_r}{2} \log \left( \frac{\min\limits_{k=1, \ldots, K} P_k}{N_r} \right)
    \end{align}
    and $o(1)$ vanishes as
    $\min\limits_{k = 1, \ldots, K} P_K \to \infty$.
  \end{corol}

\begin{remark}
    To approach the rate of \thmref{thm:proper_dim:IB:Kuser} for $K > 2$ users, 
    the number of time extensions $N$ needs to be taken to infinity, in general.
    Using a finite number of time extensions entails a loss in performance, which decays when this number is increased. 
    On the other hand, a large number of time extensions requires using a longer block, or alternatively, if the blocklength is limited, this corresponds to shortening the effective blocklength of the scalar codes used, which in turn translates to a loss in performance. 
    Hence in practice, striking a balance between these two losses is required, 
    by selecting an appropriate number of time extensions $N$.
    For further discussion and details, see~\cite{JET:SeveralUsers2015:FullPaper}.
\end{remark}

%%%%%%%%%%%%%%%%%%%%%%%%%%%%%%%%%%%%%%%%%%%%%%%%%%%%%%%%%%%%%%%%%%%%%%%%%%%%%%%%%%%%%%%%%%%%%%%%%%%%%%%%%%%%%%

\section{Application to the MIMO Two-Way Relay~Channel}
\label{s:TWRC}

In this section we apply the MIMO DMAC scheme of \secref{s:MIMO:DMAC:2users} to the MIMO two-way relay~channel (TWRC).
For Gaussian scalar channels, the similarity between these two settings was previously observed; see, e.g., \cite{ZamirBookNazerChapter}.

As in the scalar DMAC setting \eqref{eq:DoublyDirtyMAC:Model},
structured schemes based on lattice coding are known to be optimal in the limit of high SNR for the scalar TWRC \cite{WilsonRelays,NamChung_TwoWay}.
This suggests in turn that, like in the MIMO DMAC setting, using the JET in conjunction with scalar lattice codes can achieve a similar result in the MIMO TWRC setting.

We start by introducing the MIMO TWRC model.

%-----------------------------------------------------------------------------------------------------

\subsection{Channel Model}
\label{ss:TWRC:MIMO:model}

The TWRC consists of two terminals and a relay.
We define the channel model as follows. Transmission takes place in two phases, each one, without loss of generality, consisting of $n$ channel uses.
At each time instant $i$ in the first phase,
terminal $k$ ($k=1,2$) transmits a signal $x_{k;i}$ and the relay receives $y_i$ according to
some memoryless multiple-access channel (MAC) $W_\MAC(y|x_1,x_2)$.
At each time instant $i$ in the second phase,
the relay transmits a signal $x_i$ and
terminal $k$ ($k=1,2$) receives $y_{k;i}$ according to some memoryless broadcast (BC) channel $W_\BC(y_1,y_2|x)$.
Before transmission begins, terminal $k$ possesses an independent message of rate $R_k$,
unknown to the other nodes;
at the end of the two transmission phases, each terminal should be able to decode, with arbitrarily low error probability,
the message of the other terminal.
The closure of all achievable rate pairs $(R_1,R_2)$ is the capacity region of the network.

In the Gaussian MIMO setting, terminal $k$ ($k=1,2$) has $N_{t;k}$ transmit antennas and the relay has $N_r$ receive antennas, during the MAC phase (see also \figref{fig:MAC_phase}):
\begin{align}
\label{eq:TWRC:MIMO:MAC}
    \by & = \bH_1 \bx_1 + \bH_2 \bx_2 + \bz ,
\end{align}
where the channel matrices $\bH_k$ and the vectors $\bz$ are defined as in \eqref{eq:MIMO_DMAC:model}. We assume that the channel matrices are proper.\footnote{For a treatment of the case of non-proper channel matrices,
see \cite{TwoWayRelaySpaceDivision}.}
We denote by $\bK_k$ ($k=1,2$) the input covariance matrix used by terminal $k$ during transmission.

We shall concentrate on the symmetric setting:
\begin{subequations}
\label{eq:TWRC_symmetric}
\begin{align}
\label{eq:symmetric_rates}
    R_1 &= R_2 \triangleq R ,
\\
\label{eq:symmetric_powers}
  P_1 &= P_2 \triangleq P .
  \end{align}
\end{subequations}

The exact nature of the BC channel is not material in the context of this work.
We characterize it by its
common-message capacity $\Ccommon$.

\begin{figure}[t]
    \centering
    \subfloat[The MIMO MAC phase.]
    {
    \label{fig:MAC_phase}
        \psfrag{&T1}{Terminal 1}
        \psfrag{&T2}{Terminal 2}
        \psfrag{&R}{Relay}
        \psfrag{&H1}{$\bH_1$}
        \psfrag{&H2}{$\bH_2$}
        \epsfig{file = 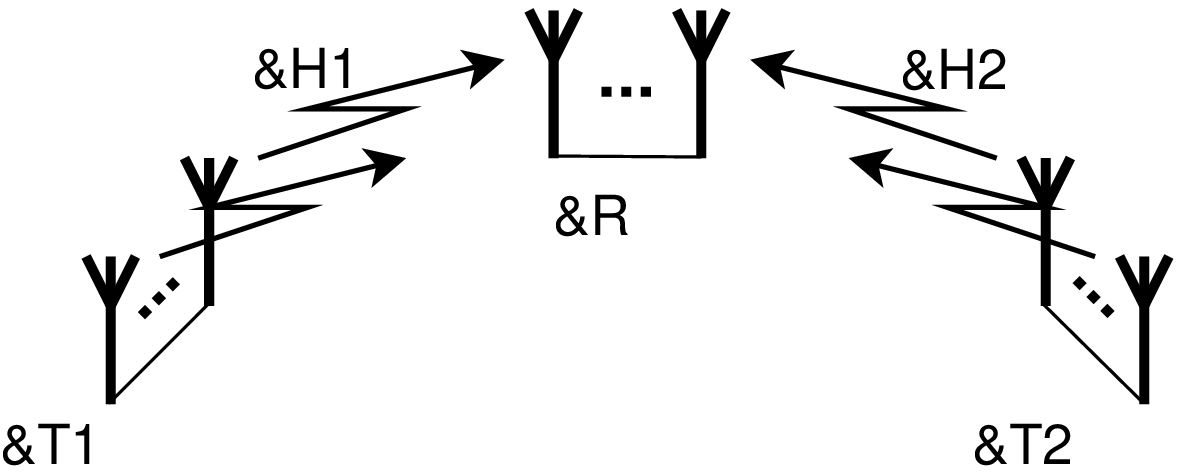, scale = .6}
    }
    \\
    \subfloat[The BC phase for the special case where it is a MIMO BC channel.]
    {
    \label{fig:BC_phase}
        \psfrag{&T1}{Terminal 1}
        \psfrag{&T2}{Terminal 2}
        \psfrag{&R}{Relay}
        \psfrag{&H1}{$\bG_1$}
        \psfrag{&H2}{$\bG_2$}
        \epsfig{file = 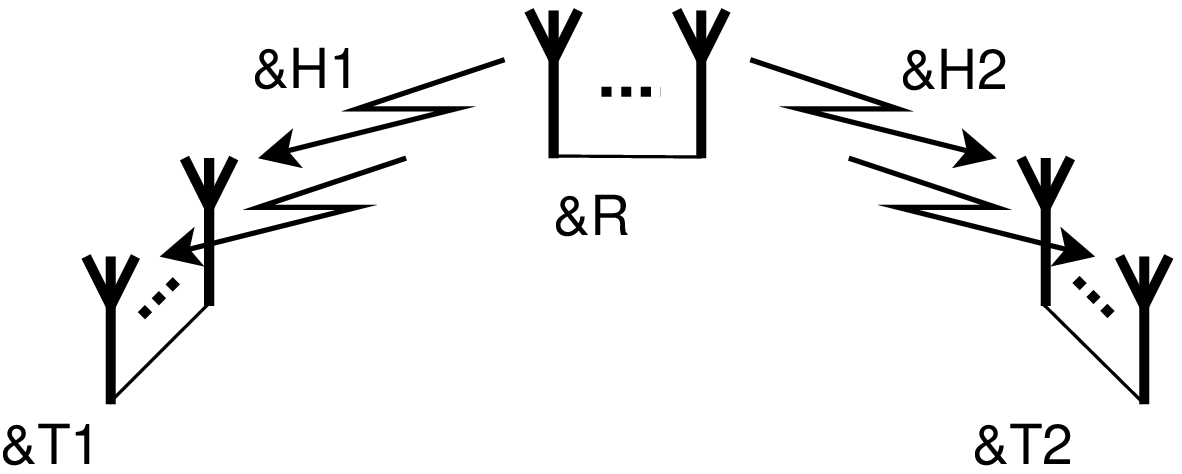, scale = .6}
    }
    \caption{The MIMO two-way relay channel. The second (BC) phase may be different in general.}
\label{fig:channel}
\end{figure}

Before treating the MIMO setting, we start by reviewing the structured physical network coding (PNC) approach for the scalar setting, in which $\bH_1$ and $\bH_2$ are replaced with scalar channel gains $h_1$ and $h_2$.

%---------------------------------------------------------------------------------------------------

\subsection{PNC for the Scalar Gaussian TWRC}
\label{ss:TWRC:SISO}

In the scalar case, \eqref{eq:DPC:MIMO:det=1} reduces to
\begin{align}
\label{eq:symmetric_gains}
    h_1 = h_2 = 1 .
\end{align}

In this case, the min-cut max-flow theorem~\cite{MinCutMaxFlow_EliasFeinsteinShannon,MinCutMaxFlow_FordFulkerson} reduces to the minimum between the individual capacities over the MAC link and the common-message capacity of the BC link (see, \eg,  \cite[Section~III]{NamChung_TwoWay}, \cite[Section~II-A]{JET:TwoWayRelay:ISITA2014}):
\begin{align}
\label{eq:SISO:symmetric:CS}
  R_\CS = \min \left\{ \frac{1}{2} \log \left(1 + P \right), \Ccommon \right\} .
\end{align}

In the (structured) PNC approach~\cite{WilsonRelays}, both terminals transmit codewords generated from the same lattice code.
Due to the linearity property of the lattice code, the sum of the two codewords is a valid lattice codeword.
This sum is decoded at the relay and sent to the terminals. Each terminal then recovers the sum codeword and subtracts from it its own lattice codeword,
to obtain the codeword transmitted by the other terminal. The rate achievable using this scheme is given by~\cite{WilsonRelays,NamChung_TwoWay}:
\begin{align}
\label{eq:SISO:symmetric:PNC}
    R_\PNC = \min \left\{ \left[ \frac{1}{2} \log \left(\frac{1}{2} + P \right) \right]^+, \
                          \Ccommon \right\} .
\end{align}

As is shown in \cite{NamChung_TwoWay}, the PNC rate \eqref{eq:SISO:symmetric:PNC} is within half a bit from the cut-set bound \eqref{eq:SISO:symmetric:CS} and approaches the cut-set bound in the limit of high SNR.
We next extend the asymptotic optimality of the PNC scheme to the MIMO setting.

We note that, in the scalar case, other strategies can also be used rather than PNC, yielding better performance at different parameter regimes.
We address these strategies in \secref{s:GSVD} along with their extension to the MIMO setting.

%---------------------------------------------------------------------------------------------------

\subsection{PNC for the Gaussian MIMO TWRC}
\label{ss:TWRC:MIMO}

The cut-set bound, in this case, is given by
\begin{align}
\label{eq:MIMO:CS}
  R_\CS = \min \left\{ \mathrm{C}_1, \mathrm{C}_2, \Ccommon \right\} ,
\end{align}
where
\begin{align}
    \mathrm{C}_k \triangleq \max_{\bK_k} \frac{1}{2} \log \left| \bI +  \bH_k \bK_k \bH_k^\dagger \right| , \quad k=1,2 \,,
\end{align}
are the individual capacities of the MIMO links, and the maximization is carried over all $\bK_k$ subject to the input power constraint $P$.

The PNC approach can be extended to work for the MIMO case by applying the JET of \secref{ss:GTD} to the channel matrices $\bH_1$ and $\bH_2$. The off-diagonal elements of the resulting triangular matrices are treated as interferences known at the transmitters. Specifically, the following rate is achievable:\footnote{Indeed, the same was proven using a different approach \cite{KAIST_TwoWay_Full}, see \secref{s:GSVD} in the sequel for a comparison.}

    \begin{align}
    \label{eq:TWRC:proper_dim:IB}
        R = \min \left\{ \frac{N_r}{2} \log \left( \frac{P}{N_r} \right) ,\ \Ccommon \right\} .
    \end{align}

  To see this, use the scheme for MIMO DMAC of \secref{s:MIMO:DMAC:2users} with the DMAC codes replaced with PNC ones: over each sub-channel, each of the users transmits a codeword from the (same) lattice,
    and the relay (who takes the role of the receiver in the scheme for the MIMO DMAC) recovers the sum codeword (modulo the coarse lattice) as in scalar PNC and continues as in the standard scalar PNC scheme.
    The achievable symmetric-rate is equal to the sum-rate achievable by the MIMO DMAC scheme \eqref{eq:proper_dim:IB:2user}, which is equal, in turn, to~\eqref{eq:TWRC:proper_dim:IB}.

The achievability of \eqref{eq:TWRC:proper_dim:IB} implies that when the channel is ``MAC-limited'' (the bottleneck being the first phase), the capacity in the high SNR regime is achieved by applying the JET to the
channel matrices and using PNC over the resulting scalar channels.

\begin{remark}[MIMO BC phase]
    The only assumption we used regarding the BC section is that it has a common-message
    capacity $\Ccommon$.
    However, it seems likely that the links to the terminal will be
    wireless MIMO ones as well. In that case, depicted also in \figref{fig:BC_phase}, the complexity may be
    considerably reduced by using a scheme that is based upon the JET, for that section
    as well; see \cite{STUD:SP}.
\end{remark}

%----------------------------------------------------------------------------------------

\subsection{Extensions and Comparisons}
\label{s:GSVD}

A different approach to the MIMO TWRC was proposed by Yang \etal~\cite{KAIST_TwoWay_Full}. In this approach, parallel scalar TWRCs are obtained via a joint matrix decomposition, but instead of the JET, the generalized singular value decomposition (GSVD)~\cite{VanLoan76, GolubVanLoan3rdEd} is used, which may be seen as a different choice of joint matrix triangularization \eqref{eq:STUD}; also, rather than DPC, successive interference cancellation (SIC) is used.

We next compare the approaches, followed by a numerical example; the reader is referred to \cite{JET:TwoWayRelay:ISITA2014} for further details.

\begin{enumerate}
\item 
    \textbf{High-SNR optimality.} In the symmetric case \eqref{eq:TWRC_symmetric}, both approaches achieve the high-SNR optimal rate \eqref{eq:TWRC:proper_dim:IB}.
\item 
    \textbf{Finite-SNR improvements.} At finite SNR, one can easily improve upon  \eqref{eq:TWRC:proper_dim:IB} to get a slightly better rate, see \remref{remark:finite_SNR}. As we show in \cite[Thm. 1]{JET:TwoWayRelay:ISITA2014}  using the technique of \cite{SuccessiveComputeForward}, the finite-SNR rate of the GSVD-based scheme can also be slightly improved.\footnote{\cite{TwoWayRelayParallelChannels} also introduces an improvement to the GSVD-based scheme for certain special cases, but it is subsumed by the improvement in \cite{JET:TwoWayRelay:ISITA2014}.} Analytically comparing the finite-SNR performance of the schemes after the improvement is difficult, although numerical evidence suggests that the JET-based scheme has better finite-SNR performance in many cases (\eg, \exref{ex:diagonal_channels}).
\item 
    \textbf{Use of any strategy.}
    For the scalar case, we note the following approaches, as an alternative to physical PNC. In the decode-and-forward (DF) approach, the relay decodes both messages with sum-rate $2R$. Instead of forwarding both messages,
    it can use a network-coding approach~\cite{NetworkCoding_Ahlswede} and sum them bitwise \mbox{modulo-2} (``\XOR\ them'').
    Then, each terminal can \XOR\ out its own message to obtain the desired one.
    In the compress-and-forward (CF) approach, the noisy sum of the signals transmitted by the sources is quantized at the relay,
    using remote Wyner--Ziv coding \cite{Yamamoto80}, with each terminal using its transmitted signal as decoder side-information.
    The CF and DF strategies both have an advantage over physical PNC for some parameters. Further, \cite{GunduzTuncelNayak_TwoWayCF} presents various layered combinations of the CF and DF approaches, which yield some improvement over the two.
    Constructing an optimized scheme via time-sharing between these approaches is pursued in~\cite{JET:TwoWayRelay:ISITA2014}. 

    Turning our attention to the MIMO setting, since the JET approach uses DPC, any strategy can be used over the subchannels; the decoder for each subchannel will receive an input signal as if this were the only channel. In contrast, the GSVD approach uses SIC, where the task of canceling inter-channel interference is left to the relay. In order to cancel out interference, the relay thus needs to decode; this prohibits the use of CF.
    \item 
    \textbf{Asymmetric case.} The GSVD-based approach is suitable for any ratio of rates and powers between the users, and in fact it is optimal in the high-SNR limit for any such ratio. 
    Since the JET-based PNC scheme uses the MIMO DMAC scheme of \secref{s:MIMO:DMAC:2users} for transmission during the MAC phase, it is limited by the minimum of the powers of the two terminals and cannot leverage excess power in case the transmit powers are different.\footnote{This statement is precise when the SNRs are high and the interferences are strong. Otherwise, some improvements are possible~\cite{PhilosofZamirErezKishti09,DMAC:FinitePowerSI:ZhuGastpar:ISIT2014}.}
    To take advantage of the excess power in case of unequal transmit powers for this scheme, one can superimpose a DF strategy for the stronger user on top of the symmetric JET-based PNC scheme, as was proposed in \cite{PNCwithSuperposition:Chung:ICC2008}. 
\end{enumerate}

\begin{example}
\label{ex:diagonal_channels}
  Consider a Gaussian MIMO TWRC with a MAC phase comprising two parallel asymmetric channels
  \begin{align}
      \bH_1 = \left(
	      \begin{array}{cc}
		1/4 & 0 \\
		0 & 4 \\
	      \end{array}
	    \right)
	    , \quad
      \bH_2 = \left(
	      \begin{array}{cc}
		4 & 0 \\
		0 & 1/4  \\
	      \end{array}
	    \right)
	    , 
  \end{align}
  and a common-message BC capacity of $\Ccommon = 20\ \text{bits}$, 
  where the terminals are subject to a per-antenna individual power constraint $P$.

  \figref{fig:parallel_channels} depicts the different achievable rates mentioned in this section as a function of $P$.
    In contrast to the case of general channel matrices, in the case of parallel channels (corresponding to diagonal channel matrices), 
    all the scalar asymmetric techniques can be used.
    Nonetheless, when considering high enough SNR, where PNC is advantageous, one observes that such asymmetric techniques are inferior to their symmetric counterparts (resulting after applying the JET).
    This gap is especially pronounced, if we compare the optimum asymmetric strategy with the optimal JET-based hybrid strategy.
\end{example}

\begin{figure}[t]
\vspace{-1.5\baselineskip}
\hspace{-.52cm}
       \includegraphics[width=1.15\columnwidth]{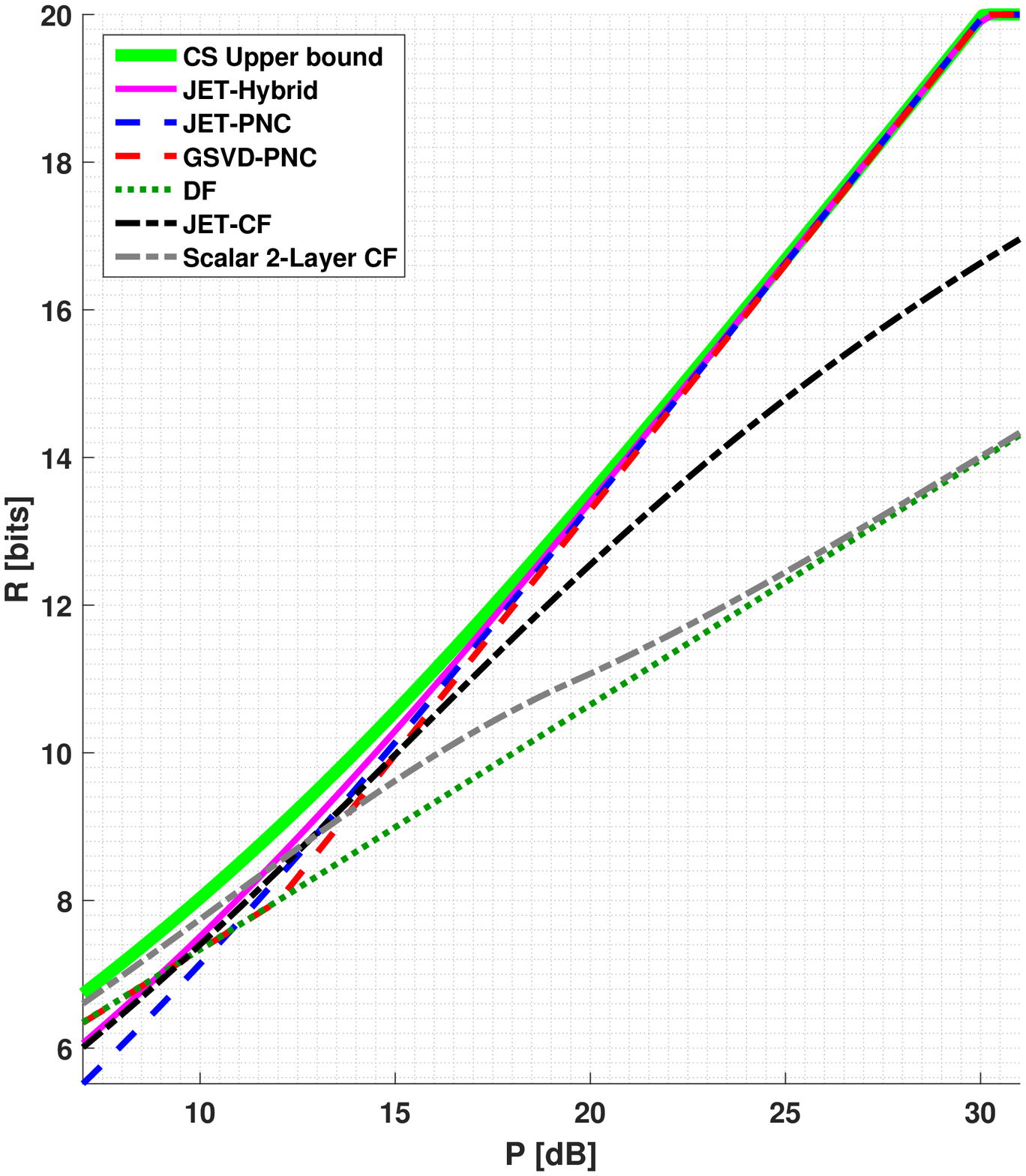}
       \caption{Performance of the proposed strategies for $\bH_1 = \diag(1/4, 4)$, $\bH_2 = \diag(4, 1/4)$ and $\Ccommon = 20 \text{bits}$.}
       \label{fig:parallel_channels}
\end{figure}

%%%%%%%%%%%%%%%%%%%%%%%%%%%%%%%%%%%%%%%%%%%%%%%%%%%%%%%%%%%%%%%%%%%%%%%%%%%%%%%%%%%%%%%%%%%%%%%

\section{Discussion: General Channel Matrices}
\label{s:discussion}

In this paper we restricted attention to full rank channel matrices having more columns than rows.
In this case, the column spaces of both matrices are equal.
Indeed, the scheme and inner bound of Sections \ref{s:MIMO:DMAC:2users} and \ref{s:MIMO:DMAC:Kusers} can be extended to work for the general case as well;
this requires, however, introducing an output projection at the receiver, which transforms the channel matrices to effective proper ones.
Since all interferences need to be canceled out for the recovery of the transmitted messages, it seems that such a scheme would be optimal in the limit of large transmit powers $P_1, \ldots, P_K \to \infty$. Unfortunately, the upper bound of \propref{prop:DMAC:2user:UB}, which is equal to the maximal individual capacity,
is not tight in the non-proper matrix case, which calls for further research.

%%%%%%%%%%%%%%%%%%%%%%%%%%%%%%%%%%%%%%%%%%%%%%%%%%%%%%%%%%%%%%%%%%%%%%%%%%%%%%%%%%%%%%%%%%%%%%%%%%%%%%%%%%%%%%

\appendix[Proof of \colref{corol:MIMO_DMAC:2user:capacity}]

The proof of \colref{corol:MIMO_DMAC:2user:capacity} is a simple adaptation of the proof of the outer bound for the scalar case \eqref{eq:SISO:DMAC:UB} of
\cite{PhilosofZamirErezKishti09}.

Take $\bs_2 \equiv 0$ and $\bs_1$ to be i.i.d.\ Gaussian with zero mean and scaled-identity covariance matrix $Q_1 \bI_{N_r}$.
Further assume that both users wish to transmit the same common message $w$, and denote the rate of this message by $R$.
Clearly, the supremum over all achievable rates $R$ bounds from the above the sum-capacity of the two-user MIMO DMAC.

By applying Fano's inequality, we have
\begin{subequations}
\label{eq:Fano}
\noeqref{eq:Fano:1,eq:Fano:2,eq:Fano:3}
\begin{align}
    n R &\leq H(w)
\label{eq:Fano:1}
 \\ &= H(w | \by^n) + I\left( w ; \by^n \right)
\label{eq:Fano:2}
 \\ &= I\left( w ; \by^n \right)  + n \eps_n ,
\label{eq:Fano:3}
\end{align}
\end{subequations}
where $\eps_n \to 0$ as the error probability goes to zero and \mbox{$n \to \infty$}.
By retracing (120)--(125) of \cite{PhilosofZamirErezKishti09} we attain
\begin{align}
\label{eq:PhilosofZamirErezKishti09}
    I\left( w ; \by^n \right) \leq h\left( \by^n \right) - h\left( \bz^n \right) - h\left( \bs_1^n \right) + h \left( \bH_1 \bx_1^n + \bz^n \right) .\ \:
\end{align}
By recalling that $\bs_1 \sim \cN\left( {\bm 0}, \bI_{N_r} \right)$ and
using the Cauchy--Schwarz inequality, we have
\begin{subequations}
\label{eq:CauchySchwarz}
\noeqref{eq:CauchySchwarz:x+z,eq:CauchySchwarz:y}
\begin{align}
    h \left( \bH_1 \bx_1^n + \bz^n \right) \leq \frac{n}{2} \log \left| \bI_{N_r} + \bH_1 \bK_1 \bH_1^\dagger \right| ,
\label{eq:CauchySchwarz:x+z}
 \\ h \left( \by^n \right) \leq \frac{n N_r}{2} \log  \left( Q_1 \right) \cdot \left( 1 +  o(1) \right),
\label{eq:CauchySchwarz:y}
\end{align}
\end{subequations}
where $o(1) \to 0$ for $Q_1 \to \infty$.

By substituting \eqref{eq:CauchySchwarz} in \eqref{eq:PhilosofZamirErezKishti09} and the outcome in \eqref{eq:Fano:3},
and taking $Q_1 \to \infty$, we attain
\begin{align}
    R \leq \frac{1}{2} \log \left| \bI_{N_r} + \bH_1 \bK_1 \bH_1^\dagger \right| + \eps_n .
\end{align}

By switching roles between the users, the following upper bound holds
\begin{align}
    R \leq \frac{1}{2} \log \left| \bI_{N_r} + \bH_2 \bK_2 \bH_2^\dagger \right| + \eps_n ,
\end{align}
and the desired result follows.
\hfill $\blacksquare$

%%%%%%%%%%%%%%%%%%%%%%%%%%%%%%%%%%%%%%%%%%%%%%%%%%%%%%%%%%%%%%%%%%%%%%%%%%%%%%%%%%%%%%%%%%%%%%%%%%%%%%%%%%%%%%


\bibliographystyle{IEEEtran}

\begin{thebibliography}{10}
\providecommand{\url}[1]{#1}
\csname url@samestyle\endcsname
\providecommand{\newblock}{\relax}
\providecommand{\bibinfo}[2]{#2}
\providecommand{\BIBentrySTDinterwordspacing}{\spaceskip=0pt\relax}
\providecommand{\BIBentryALTinterwordstretchfactor}{4}
\providecommand{\BIBentryALTinterwordspacing}{\spaceskip=\fontdimen2\font plus
\BIBentryALTinterwordstretchfactor\fontdimen3\font minus
  \fontdimen4\font\relax}
\providecommand{\BIBforeignlanguage}[2]{{%
\expandafter\ifx\csname l@#1\endcsname\relax
\typeout{** WARNING: IEEEtran.bst: No hyphenation pattern has been}%
\typeout{** loaded for the language `#1'. Using the pattern for}%
\typeout{** the default language instead.}%
\else
\language=\csname l@#1\endcsname
\fi
#2}}
\providecommand{\BIBdecl}{\relax}
\BIBdecl

\bibitem{Costa83}
M.~H.~M. Costa, ``Writing on dirty paper,'' \emph{IEEE Trans.\ Inf.\ Theory},
  vol.~29, no.~3, pp. 439--441, May 1983.

\bibitem{CohenLapidoth02}
A.~S. Cohen and A.~Lapidoth, ``The {Gaussian} watermarking game,'' \emph{IEEE
  Trans.\ Inf.\ Theory}, vol.~48, no.~6, pp. 1639--1667, Jun. 2002.

\bibitem{ErezShamaiZamir05}
U.~Erez, S.~Shamai, and R.~Zamir, ``Capacity and lattice strategies for
  canceling known interference,'' \emph{IEEE Trans.\ Inf.\ Theory}, vol.~51,
  no.~11, pp. 3820--3833, Nov. 2005.

\bibitem{ZamirBookNazerChapter}
B.~Nazer and R.~Zamir, ``Gaussian networks,'' in R.~Zamir, \emph{Lattice coding for
  signals and networks}.
  Cambridge: Cambridge University Press, 2014.

\bibitem{GinisCioffiAsilomar}
J.~M. Cioffi and G.~Ginis, ``A multi-user precoding scheme achieving crosstalk
  cancellation with application to {DSL} systems,'' in \emph{Proc.\ Asilomar
  Conf.\ Sig., Sys and Comp.}, vol.~2, Pacific Grove, CA, USA, Oct./Nov.~2000,
  pp. 1627--1631.

\bibitem{CaireShamai03}
G.~Caire and S.~Shamai, ``On the achievable throughput of a multi-antenna
  {G}aussian broadcast channel,'' \emph{IEEE Trans.\ Inf.\ Theory}, vol.~49,
  no.~7, pp. 1649--1706, July 2003.

\bibitem{YuCioffiSumCapacity}
W.~Yu and J.~M. Cioffi, ``Sum capacity of {G}aussian vector broadcast
  channels,'' \emph{IEEE Trans.\ Inf.\ Theory}, vol.~50, no.~9, pp. 1875--1892,
  Sep.\ 2004.

\bibitem{WSS06}
H.~Weingarten, Y.~Steinberg, and S.~Shamai, ``The capacity region of the
  {Gaussian} multiple-input multiple-output broadcast channel,'' \emph{IEEE
  Trans.\ Inf.\ Theory}, vol.~52, no.~9, pp. 3936--3964, Sep. 2006.

\bibitem{UCD}
Y.~Jiang, W.~Hager, and J.~Li, ``Uniform channel decomposition for {MIMO}
  communications,'' \emph{IEEE Trans. Sig.\ Proc.}, vol.~53, no.~11, pp.
  4283--4294, Nov. 2005.

\bibitem{PhilosofZamirErezKishti09}
T.~Philosof, R.~Zamir, U.~Erez, and A.~Khisti, ``Lattice strategies for the
  dirty multiple access channel,'' \emph{IEEE Trans.\ Inf.\ Theory}, vol.~57,
  no.~8, pp. 5006--5035, Aug. 2011.

\bibitem{PhilosofZamir09}
T.~Philosof and R.~Zamir, ``On the loss of single-letter characterization: The
  dirty multiple access channel,'' \emph{IEEE Trans.\ Inf.\ Theory}, vol.~55,
  no.~6, pp. 2442--2454, June 2009.

\bibitem{Telatar99}
I.~E. Telatar, ``Capacity of the multiple antenna {G}aussian channel,''
  \emph{Europ. Trans. Telecommun.}, vol.~10, no.~6, pp. 585--595, Nov. 1999.

\bibitem{Foschini96}
G.~Foschini, ``Layered space--time architecture for wireless communication in a
  fading environment when using multi-element antennas,'' \emph{Bell Sys. Tech.
  Jour.}, vol.~1, no.~2, pp. 41--59, 1996.

\bibitem{CioffiForneyGDFE}
J.~M. Cioffi and G.~D. {Forney Jr.}, ``Generalized decision-feedback
  equalization for packet transmission with {ISI} and {G}aussian noise,'' in
  \emph{Comm., Comp., Cont. and Sig. Proc.}\hskip 1em plus 0.5em minus
  0.4em\relax US: Springer, 1997, pp. 79--127.

\bibitem{Wolniansky_V-BLAST}
P.~W. Wolniansky, G.~J. Foschini, G.~D. Golden, and R.~A. Valenzuela,
  ``{V-BLAST}: An architecture for realizing very high data rates over the
  rich-scattering wireless channel,'' in \emph{Proc.\ URSI Int.\ Symp.\ Sig.,
  Sys., Elect.\ (ISSSE)}, Sep./Oct.\ 1998, pp. 295--300.

\bibitem{HassibiVBLAST}
B.~Hassibi, ``An efficient square-root algorithm for {BLAST},'' in \emph{Proc.
  IEEE Int. Conf.\ Acoust.\, Speech and Sig.\ Proc. (ICASSP)}, vol.~2,
  Istanbul, Turkey, June 2000, pp. 737--740.

\bibitem{STUD:SP}
A.~Khina, Y.~Kochman, and U.~Erez, ``Joint unitary triangularization for {MIMO}
  networks,'' \emph{IEEE Trans. Sig.\ Proc.}, vol.~60, no.~1, pp. 326--336,
  Jan. 2012.

\bibitem{WeylCondition}
H.~Weyl, ``Inequalities between two kinds of eigenvalues of a linear
  transformation,'' in \emph{Proc. Nat. Acad. Sci. USA, 35}, no.~7, May 1949,
  pp. 408--411.

\bibitem{WeylConditionInverse_ByHorn}
A.~Horn, ``On the eigenvalues of a matrix with prescribed singular values,'' in
  \emph{Proc. Amer. Math. Soc.}, vol.~5, no.~1, Feb. 1954, pp. 4--7.

\bibitem{GTD}
Y.~Jiang, W.~Hager, and J.~Li, ``The generalized triangular decompostion,''
  \emph{Math.\ of Comput.}, vol.~77, no. 262, pp. 1037--1056, Oct. 2008.

\bibitem{QRS-GTD}
\BIBentryALTinterwordspacing
J.-K. Zhang and K.~M. Wong, ``Fast {QRS} decomposition of matrix and its
  applications to numerical optimization,'' Dpt.\ of Elect.\ and Comp.
  Engineering, McMaster University, Tech. Rep. [Online]. Available:
  \url{http://www.ece.mcmaster.ca/~jkzhang/papers/sam_qrs.pdf}
\BIBentrySTDinterwordspacing

\bibitem{UnityTriangularization}
P.~Kosowski and A.~Smoktunowicz, ``On constructing unit triangular matrices
  with prescribed singular values,'' \emph{Computing}, vol.~64, no.~3, pp.
  279--285, May 2000.

\bibitem{JET:SeveralUsers2015:FullPaper}
A.~Khina, I.~Livni, A.~Hitron, and U.~Erez, ``Joint unitary triangularization
  for {G}aussian multi-user {MIMO} networks,'' \emph{IEEE Trans.\ Inf.\
  Theory}, vol.~61, no.~5, pp. 2662--2692, May 2015.

\bibitem{GolubVanLoan3rdEd}
G.~H. Golub and C.~F. {Van Loan}, \emph{Matrix Computations, {3rd ed.}}\hskip
  1em plus 0.5em minus 0.4em\relax Baltimore: Johns Hopkins University Press,
  1996.

\bibitem{QRS}
J.-K. Zhang, A.~Kav\v{c}i\'{c}, and K.~M. Wong, ``Equal-diagonal {QR}
  decomposition and its application to precoder design for
  successive-cancellation detection,'' \emph{IEEE Trans.\ Inf.\ Theory},
  vol.~51, no.~1, pp. 154--172, Jan. 2005.

\bibitem{GMD}
Y.~Jiang, W.~Hager, and J.~Li, ``The geometric mean decompostion,'' \emph{Lin.\
  Algebra and Its Apps.}, vol. 396, pp. 373--384, Feb. 2005.

\bibitem{MartinianWaterfilling}
\BIBentryALTinterwordspacing
E.~Martinian, ``Waterfilling gains at most {O(1/SNR)} at high {SNR},'' Feb.
  2004. [Online]. Available:
  \url{http://www.rle.mit.edu/sia/wp-content/uploads/2015/04/2004-martinian-unpublished.pdf}
\BIBentrySTDinterwordspacing

\bibitem{HornJohnsonBook}
R.~A. Horn and C.~R. Johnson, \emph{Topics in Matrix Analysis}.\hskip 1em plus
  0.5em minus 0.4em\relax Cambridge: Cambridge University Press, 1991.

\bibitem{WilsonRelays}
M.~P. Wilson, K.~Narayanan, H.~Pfister, and A.~Sprintson, ``Joint physical
  layer coding and network coding for bidirectional relaying,'' \emph{IEEE
  Trans.\ Inf.\ Theory}, vol.~56, pp. 5641--5654, Nov.~2010.

\bibitem{NamChung_TwoWay}
W.~Nam, S.-Y. Chung, and Y.~H. Lee, ``Capacity of the {G}aussian two-way relay
  channel to within 1/2 bit,'' \emph{IEEE Trans.\ Inf.\ Theory}, vol.~56,
  no.~11, pp. 5488--5494, Nov. 2010.

\bibitem{TwoWayRelaySpaceDivision}
X.~Yuan, T.~Yang, and I.~B. Collings, ``Multiple-input multiple-output two-way
  relaying: a space--division approach,'' \emph{IEEE Trans.\ Inf.\ Theory},
  vol.~59, no.~10, pp. 6421--6440, Oct. 2013.

\bibitem{MinCutMaxFlow_EliasFeinsteinShannon}
P.~Elias, A.~Feinstein, and C.~E. Shannon, ``A note on the maximum flow through
  a network,'' \emph{Proc.\ IRE}, vol.~2, no.~4, pp. 117--119, 1956.

\bibitem{MinCutMaxFlow_FordFulkerson}
L.~R. {Ford~Jr.} and D.~R. Fulkerson, ``Maximal flow through a network,''
  \emph{Canad.\ J.\ Math.}, vol.~8, no.~3, pp. 399--404, 1956.

\bibitem{JET:TwoWayRelay:ISITA2014}
A.~Khina, Y.~Kochman, and U.~Erez, ``Improved rates and coding for the {MIMO}
  two-way relay channel,'' in \emph{Proc.\ IEEE Int.\ Symp.\ Info.\ Theory and
  Its Apps.\ (ISITA)}, Melbourne, Vic, Australia, Oct. 2014, pp. 658--662.

\bibitem{KAIST_TwoWay_Full}
H.~J. Yang, C.~Joohwan, and A.~Paulraj, ``Asymptotic capacity of the separated
  {MIMO} two-way relay channel,'' \emph{IEEE Trans.\ Inf.\ Theory}, vol.~57,
  no.~11, pp. 7542--7554, Nov. 2011.

\bibitem{VanLoan76}
C.~F. {Van Loan}, ``Generalizing the singular value decomposition,'' \emph{SIAM
  J. Numer.}, vol.~13, no.~1, pp. 76--83, Mar. 1976.

\bibitem{SuccessiveComputeForward}
B.~Nazer, ``Successive compute-and-forward,'' in \emph{Proc.\ Biennial Int.\
  Zurich Seminar on Comm.}, Zurich, Switzerland, March 2012, pp. 103--106.

\bibitem{TwoWayRelayParallelChannels}
Y.-C. Huang, K.~Narayanan, and T.~Liu, ``Coding for parallel {Gaussian}
  bi-directional relay channels: A deterministic approach,'' in \emph{Proc.\
  Annual Allerton Conf.\ on Comm., Control, and Comput.}, Monticello, IL, USA,
  Sep.2011, pp. 400--407.

\bibitem{NetworkCoding_Ahlswede}
R.~Ahlsewede, N.~Cai, S.-Y. Li, and R.~Yeung, ``Network information flow,''
  \emph{IEEE Trans.\ Inf.\ Theory}, vol.~46, no.~4, pp. 1204--1216, July 2000.

\bibitem{Yamamoto80}
H.~Yamamoto and K.~Itoh, ``Source coding theory for multiterminal communication
  systems with a remote source,'' \emph{Tran.\ IECE of Japan}, vol. E\ 63,
  no.~10, pp. 700--706, 1980.

\bibitem{GunduzTuncelNayak_TwoWayCF}
D.~G{\"u}nd{\"u}z, E.~Tuncel, and J.~Nayak, ``Rate regions for the separated
  two-way relay channel,'' in \emph{Proc.\ Annual Allerton Conf.\ on Comm.,
  Control, and Comput.}, Monticello, IL, USA, Sep.\ 2010, pp. 1333--1340.

\bibitem{DMAC:FinitePowerSI:ZhuGastpar:ISIT2014}
J.~Zhu and M.~Gastpar, ``Gaussian (dirty) multiple access channels: a
  compute-and-forward prespective,'' in \emph{Proc. IEEE Int. Symp.\ on Inf.\
  Theory (ISIT)}, Honolulu, HI, USA, June/July 2014, pp. 2949--2953.

\bibitem{PNCwithSuperposition:Chung:ICC2008}
I.-J. Baik and S.-Y. Chung, ``Network coding for two-way relay channels using
  lattices,'' in \emph{Proc. IEEE Int. Conf.\ on Comm.\ (ICC)}, Beijing, China,
  May 2008, pp. 3898--3902.

\end{thebibliography}

%%%%%%%%%%%%%%%%%%%%%%%%%%%%%%%%%%%%%%%%%%%%%%%%%%%%%%%%%%%%%%%%%%%%%%%%%%%%%%%%%%%%%%%%%%%%%%%%%%%%%%%%%%%%%%%%%%%%%%%%%

\begin{IEEEbiographynophoto}{Anatoly Khina}
    (S'08)
    was born in Moscow, USSR, on September 10, 1984. He received the B.Sc. (\emph{summa cum laude}), M.Sc.\ (\emph{summa cum laude}) and Ph.D. degrees from Tel Aviv University, Tel-Aviv, Israel in 2006, 2010 and 2016, respectively, all in electrical engineering. 
    He is currently a Postdoctoral Scholar in the Department of Electrical Engineering at the California Institute of Technology, Pasadena, CA, USA.
    His research interests include information theory, control theory, signal processing and matrix analysis. 

    In parallel to his studies, Dr. Khina had been working as an engineer in various algorithms, software and hardware R\&D positions. He is a recipient of the Fulbright, Rothschild and Marie Sk\l odowska-Curie Postdoctoral Fellowships, Clore Scholarship, Trotsky Award, Weinstein Prize in signal processing, Intel award for Ph.D. research, and the first prize for outstanding research work in the field of communication technologies of the Advanced Communication Center (ACC) Feder Family Award.
\end{IEEEbiographynophoto}

\begin{IEEEbiographynophoto}{Yuval Kochman}
    (S'06--M'09) 
    received his B.Sc.\ (\emph{cum laude}), M.Sc.\ (\emph{cum laude}) and Ph.D.\ degrees from Tel Aviv University in 1993, 2003 and 2010, respectively, all in electrical engineering. 
    During 2009--2011, he was a Postdoctoral Associate at the Signals, Informtion and Algorithms Laboratory at the Massachusetts Institute of Technology (MIT), Cambridge, MA, USA. 
    Since 2012, he has been with the School of Computer Science and Engineering at the Hebrew University of Jerusalem. 
    Outside academia, he has worked in the areas of radar and digital communications. His research interests include information theory, communications and signal processing.
\end{IEEEbiographynophoto}

\begin{IEEEbiographynophoto}{Uri Erez}
    (M'09) 
    was born in Tel-Aviv, Israel, on October 27, 1971.
    He received the B.Sc.\ degree in mathematics and physics and the M.Sc.\ and
    Ph.D.\ degrees in electrical engineering from Tel-Aviv University in 1996,
    1999, and 2003, respectively. 
    During 2003--2004, he was a Postdoctoral Associate at the Signals, Information and Algorithms Laboratory at the
    Massachusetts Institute of Technology (MIT), Cambridge, MA, USA. 
    Since 2005, he has been with the Department of Electrical Engineering--Systems at Tel-Aviv
    University. 
    His research interests are in the general areas of information theory and digital communications. 
    He served in the years 2009--2011 as Associate Editor for Coding Techniques for the 
    {\sc IEEE Transactions on Information Theory}.
\end{IEEEbiographynophoto}

%%%%%%%%%%%%%%%%%%%%%%%%%%%%%%%%%%%%%%%%%%%%%%%%%%%%%%%%%%%%%%%%%%%%%%%%%%%%%%%%%%%%%%%%%%%%%%%%%%%%%%%%%%%%%%%%%%%%%%%%%
\end{document}